\documentclass[letterpaper,conference,twocolumn,10pt,oneside]{IEEEtran}
\usepackage[latin1]{inputenc}
\usepackage{amsfonts}
\usepackage{amsmath}
\usepackage{amssymb}
\usepackage[american]{babel}
\usepackage{caption2}
\usepackage[T1]{fontenc}
\usepackage[dvips]{graphicx,epsfig}
\usepackage{graphicx}
\usepackage{mathrsfs}
\usepackage{color} 
\usepackage{caption2}
\usepackage{cite}
\usepackage{balance}

\graphicspath{{Matlab/}}

\def\hksqrt{\mathpalette\DHLhksqrt}
\def\DHLhksqrt#1#2{\setbox0=\hbox{$#1\sqrt{#2\,}$}\dimen0=\ht0
\advance\dimen0-0.2\ht0
\setbox2=\hbox{\vrule height\ht0 depth -\dimen0}%
{\box0\lower0.4pt\box2}}
\newcommand{\hsqrt}[1]{\hksqrt{#1}}
\newcommand{\BibPath}{/home/milan/Documents/University/Research/BibTeX}

\newcommand{\Rl}{\mathbb{R}}

\newcommand{\eq}{\triangleq}

\newcommand{\eexp}[1]{\textrm{e}^{#1}}
\newcommand{\ejw}{(\eexp{j\w})}
\newcommand{\sigsq}{\sigma^{2}}

\newcommand{\w}{\omega}

\newcommand{\Q}{\mathcal{Q}}

\newcommand{\abs}[1]{\left| #1\right|}
\newcommand{\set}[1]{\{ #1 \}}

\newcommand{\sumfromto}[2]{\sum\nolimits_{#1}^{#2}}

\newcommand{\Intfromto}[2]{\int\limits_{#1}^{#2}}

\newcommand{\intpipi}[1]{\frac{1}{2\pi}\!\!\int_{-\pi}^{\pi} {#1} d\w}

\newcommand{\expo}[1]{\textrm{e}^{#1}}

\newcommand{\pipi}{[-\pi,\pi]}

\newcommand{\forallwinpipi}{\, \forall \w \in [-\pi,\pi]}

\newcommand{\tr}{\textrm{tr}}

\newcommand{\bA}{\boldsymbol{A}}

\newcommand{\bD}{\boldsymbol{D}}

\newcommand{\bF}{\boldsymbol{F}}

\newcommand{\bI}{\boldsymbol{I}}

\newcommand{\bK}{\boldsymbol{K}}
\newcommand{\bL}{\boldsymbol{L}}
\newcommand{\bM}{\boldsymbol{M}}

\newcommand{\bT}{\boldsymbol{T}}

\newcommand{\Irate}[2]{\bar{I}(\set{#1_{k}};\set{#2_{k}})}

\newtheorem{thm}{Theorem}

\newtheorem{lem}{Lemma}
\newtheorem{rem}{Remark}

\begin{document}

\title{Achieving the Quadratic Gaussian Rate-Distortion Function for Source Uncorrelated Distortions}
\author{
\authorblockN{Milan S. Derpich,
Jan {\O}stergaard, and
Daniel E. Quevedo}
\authorblockA{School of Electrical Engineering and Computer Science, The University of Newcastle, NSW 2308, Australia\\
\authorblockA{ milan.derpich@studentmail.newcastle.edu.au, jan.ostergaard@newcastle.edu.au,
  dquevedo@ieee.org
}
}
}

\maketitle
\begin{abstract}
We prove achievability of the recently characterized quadratic Gaussian rate-distortion function (RDF) subject to the constraint that the distortion is uncorrelated to the source. This result is based on shaped dithered lattice quantization in the limit as the lattice dimension tends to infinity and holds for all positive distortions.
It turns out that this uncorrelated distortion RDF can be realized causally. This feature, which stands in contrast to Shannon's RDF, is illustrated by causal transform coding. 
Moreover, we prove that by using feedback noise shaping the uncorrelated distortion RDF can be achieved causally and with memoryless entropy coding.
Whilst achievability relies upon infinite dimensional quantizers, we prove that the rate loss incurred in the finite dimensional case can be upper-bounded by the space filling loss of the quantizer and, thus, is at most 0.254 bit/dimension. 
%

  
\end{abstract}
%
%
\section{Introduction}\label{sec:intro}
Shannon's rate-distortion function $R(D)$ for a stationary zero-mean Gaussian source $X$ with memory and under the MSE fidelity criterion can be written in a parametric form (the reverse water-filling solution)~\cite{gallag68}
\begin{subequations}\label{eq:ShannonsRDF}
\begin{align}
 R(D)&= \frac{1}{2\pi}\int_{\omega: S_X(\omega)>\theta} \frac{1}{2}\log\left(\frac{S_X(\omega)}{\theta}\right)\\
 D &= \frac{1}{2\pi}\int_{-\pi}^{\pi} S_Z(\omega) \, d\omega,\label{eq:D_shannon}
\end{align}
where $S_X(\omega)$ denotes the \emph{power spectral density} (PSD) of $X$ and the distortion PSD $S_Z(\omega)$ is given by
%
\begin{equation}
S_Z(\omega) = \begin{cases}
\theta, & \text{if}\ S_X(\omega) > \theta \\
S_X(\omega), & \text{otherwise}.
\end{cases}
\end{equation}
\end{subequations}
The water level $\theta$ is chosen such that the distortion constraint~\eqref{eq:D_shannon} is satisfied.

It is well known that in order to achieve Shannon's RDF in the quadratic Gaussian case, the distortion must be independent of the output.
This clearly implies that the distortion must be \emph{correlated} to the source.

Interestingly, many well known source coding schemes actually lead, by construction, to source-uncorrelated distortions. %
In particular, this is the case when the source coder satisfies the following two conditions: 
a) The linear processing stages (if any) achieve \emph{perfect reconstruction} (PR) in the absence of quantization; 
b) the quantization error is uncorrelated to the source.
The first condition is typically satisfied by PR filterbanks~\cite{vaidya93}, transform coders~\cite{goyal01} and feedback quantizers~\cite{jaynol84}.
The second condition is met when subtractive (and often when non-subtractive) dither quantizers are employed~\cite{grasto93}.
Thus, any PR scheme using, for example, subtractively dithered quantization, leads to source-uncorrelated distortions.

An important fundamental question, which was raised by the authors in a recent paper~\cite{derost08}, is: ``What is the impact on Shannon's rate-distortion function, when we further impose the constraint that the end-to-end distortion must be uncorrelated to the input?''

In~\cite{derost08}, we formalized the notion of $R^\perp(D)$, which is the quadratic rate-distortion function subject to the constraint that the distortion is uncorrelated to the input. 
For a Gaussian  source $X\in\Rl^{N}$, we defined $R^{\perp}(D)$ as~\cite{derost08}
\begin{equation}
R^{\perp}(D) \eq \min_{%
\substack{Y: \mathbb{E}[X(Y-X)^T]=\boldsymbol{0}, \\ 
\frac{1}{N} tr(\boldsymbol{K}_{Y-X}) \leq D,  
\frac{1}{N}|\boldsymbol{K}_{Y-X}|^{\frac{1}N} > 0}}
\tfrac{1}{N}
I(X ; Y),
\end{equation}
where the notation $\bK_{X}$ denotes the covariance matrix of $X$ and $\abs{\cdot}$ refers to the determinant.
For zero mean Gaussian stationary sources, we showed in~\cite{derost08} that the above minimum (in the limit when $N\to\infty$)
satisfies the following equations:
%
\begin{subequations}\label{eq:Rperp_equations}
\begin{align}
R^\perp(D) 
&= 
\frac{1}{2\pi} \!\int_{-\pi}^{\pi} \!\!\log\!\left(\!\!\frac{\hsqrt{S_X(\omega)+\alpha} + \hsqrt{S_X(\omega)}}{\hsqrt{\alpha}}\!\right)  d\omega\label{eq:RperpProcDef}\\
D &= \frac{1}{2\pi} \int_{-\pi}^{\pi} S_{Z}(\w)  d\omega,\nonumber
\end{align}
\text{where}
\begin{equation}\label{eq:Sz_def}
S_Z(\omega) \!= \!\frac{1}{2}\!\left(\! \hsqrt{S_X(\omega)\!+\!\alpha} - \hsqrt{S_X(\omega)}\right)\!\hsqrt{S_X(\omega)}, \; \forall \w,
\end{equation}
\end{subequations}
is the PSD of the optimal distortion, which needs to be Gaussian.
Notice that here the parameter $\alpha$ (akin to $\theta$ in~\eqref{eq:ShannonsRDF}) does not represent a ``water level''.
Indeed, unless $X$ is white, the PSD of the optimal distortion for $R^{\perp}(D)$ is not white, \emph{for all $D>0$}.
\footnote{Other similarities and differences between $R^{\perp}(D)$ and Shannon's $R(D)$ are discussed in~\cite{derost08}.}

In the present paper we prove achievability of $R^\perp(D)$ by constructing coding schemes based on dithered lattice quantization, which, in the limit as the quantizer dimension approaches infinity, are able to achieve $R^\perp(D)$ for any positive $D$. 
We also show that $R^{\perp}(D)$ can be realized causally, i.e.,
that for all Gaussian sources and for all positive distortions one can build forward test channels that realize $R^{\perp}(D)$ without using non-causal filters.
This is contrary to the case of Shannon's rate distortion function $R(D)$, where at least one of the filters of the forward test channel that realizes $R(D)$ needs to be non-causal~\cite{gallag68}.
To further illustrate the causality of
$R^{\perp}(D)$,
we present a causal transform coding architecture that realizes it.
We also show that the use of feedback noise-shaping allows one to achieve $R^{\perp}(D)$ with memoryless entropy coding. 
This parallels a recent result by Zamir, Kochman and Erez for $R(D)$~\cite{zamkoc08}.
We conclude the paper by showing that, in all the discussed architectures, the
rate-loss (with respect to $R^{\perp}(D)$) when using a finite-dimensional quantizer can be upper bounded by the space-filling loss of the quantizer.
Thus, for any Gaussian source with memory, by using noise-shaping and scalar dithered quantization,
the \emph{scalar}  entropy (conditioned to the dither) of the quantized output exceeds $R^{\perp}(D)$ by at most 0.254 bit/dimension.


\section{Background on Dithered Lattice Quantization}\label{sec:background}
A randomized lattice quantizer is a lattice quantizer with subtractive dither $\nu$, followed by entropy encoding. 
The dither $\nu \sim \mathcal{U}(V_0)$ is uniformly distributed over a Voronoi cell $V_0$ of the lattice quantizer.
Due to the dither, the quantization error is truly independent of the input. Furthermore, it was shown in~\cite{zamfed92} that the coding rate of the quantizer, i.e.\
\begin{equation}
R_{\Q_{N}} \eq  
\tfrac{1}{N}H(\Q_{N}(X+\nu)|\nu)
\end{equation}
can be written as the mutual information between the input and the output of an additive noise channel $Y'=X+E'$, where $E'$ denotes the channel's additive noise and is distributed as $-\nu$. 
More precisely, 
$R_{\Q_{N}}=\frac{1}{N}I(X;Y')=\frac{1}{N}I(X;X+E')$ and the quadratic distortion per dimension is given by $\frac{1}{N}\mathbb{E}\|Y'-X\|^2 = \frac{1}{N}\mathbb{E}\|E'\|^2$.

It has furthermore been shown that when $\nu$ is white 
there exists a sequence of lattice quantizers $\{\Q_{N}\}$ where the quantization error (and therefore also the dither) tends to be approximately Gaussian distributed (in the divergence sense) for large $N$. 
Specifically, let $E'$ have a probability distribution (PDF) $f_{E'}$, and let $E'_G$ be Gaussian distributed with the same mean and covariance as $E'$. 
Then $\lim_{N\rightarrow\infty}\frac{1}{N}D(f_{E'}(e) \| f_{E'_G}(e)) = 0$ with a convergence rate of $\frac{\log(N)}{N}$ if the sequence $\{\Q_{N}\}$ is chosen appropriately~\cite{zamfed96}.

In the next section we will be interested in the case where the dither is not necessarily white. 
By shaping the Voronoi cells of a lattice quantizer whose dither $\nu$ is white, we also shape $\nu$, obtaining a colored dither $\nu'$.
This situation was considered in detail in~\cite{zamfed96} from where we obtain the following lemma (which was proven in~\cite{zamfed96} but not put into a lemma).
\begin{lem}\label{lem:shapedlattice}
Let $E\sim \mathcal{U}(V_0)$ be white, i.e.\ $E$ is uniformly distributed over the Voronoi cell $V_0$ of the lattice quantizer $\Q_{N}$ and $\boldsymbol{K}_E=\epsilon \bI$. 
Furthermore, let $E'\sim \mathcal{U}(V'_0)$, where 
$V'_0$ denotes the shaped Voronoi cell
$V'_0 = \{x\in \mathbb{R} : \boldsymbol{M}^{-1}x \in V_0\}$ and $\bM$ is some invertible linear transformation.
Denote the covariance of $E'$ by $\boldsymbol{K}_{E'} = \boldsymbol{M}\boldsymbol{M}^T\epsilon$.  
Similarly, let $E_G\sim \mathcal{N}(\boldsymbol{0},\boldsymbol{K}_{E_G})$ having covariance matrix
$\boldsymbol{K}_{E_G} = \bK_{E}$ and let $E'_G\sim \mathcal{N}(\boldsymbol{0},\boldsymbol{K}_{E'_G})$ where $\boldsymbol{K}_{E'_G} = \boldsymbol{K}_{E'}$. 
Then there exists a sequence of shaped lattice quantizers such that
%
\begin{equation}
\tfrac{1}{N}D(f_{E'}(e) \| f_{E'_G}(e)) = \mathcal{O}\left({\log(N)}/{N}\right).
\end{equation}
%
\end{lem}
\begin{proof}
The divergence is invariant to invertible transformations since 
$h(E')=h(E)+\log_2(|\boldsymbol{M}|)$. Thus, $D(f_{E'}(e) \| f_{E'_G}(e)) = D(f_{\boldsymbol{M}E}(e)\|f_{\boldsymbol{M}E_G}(e)) = D(f_{E}(e)\|f_{E_G}(e))$ for any $N$. 
\end{proof}

\section{Achievability of $R^{\perp}(D)$}
The simplest forward channel that realizes $R^{\perp}(D)$ is shown in Fig.~\ref{fig:forwartdtc}.
According to~\eqref{eq:Rperp_equations}, all that is needed for the mutual information per dimension between $X$ and $Y$ to equal $R^{\perp}(D)$ is that $Z$ be Gaussian with PSD equal to the right hand side (RHS) of~\eqref{eq:Sz_def}.
\begin{figure}[htp]
 \centering
 \input{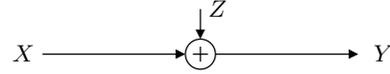}
 \caption{Forward test channel}
 \label{fig:forwartdtc}
\end{figure}

In view of the asymptotic properties of randomized lattice quantizers discussed in Section~\ref{sec:background},
the achievability of~$R^{\perp}(D)$ can be shown by replacing the test channel of Fig.\ref{fig:forwartdtc} by an adequately \emph{shaped} $N$-dimensional randomized lattice quantizer $\Q_{N}'$ and then letting $N\rightarrow \infty$.
In order to establish this result, the following lemma is needed.
\begin{lem}\label{lem:excessrate}
\emph{
Let $X$, $X'$, $Z$ and $Z'$  be mutually independent random vectors.
Let $X'$ and $Z'$ be arbitrarily distributed, and let $X$ and $Z$ be Gaussian having the same mean and covariance as $X'$ and $Z'$, respectively.
Then
%
\begin{align}
 I(X';X'+Z') \leq I(X;X+Z) + D(Z'\Vert Z).
\end{align}
}
\end{lem}

\begin{proof}
\begin{align*}
 I(&X';X'+Z') 
\overset{\hphantom{(a)}}{=} h(X'+Z') - h(Z')\\
&\overset{(a)}{=} h(X+Z) - h(Z) + D(Z'\Vert Z) - D(X'+Z'\Vert X + Z)\\
&\overset{\hphantom{(b)}} {\leq} I(X;X+Z) + D(Z'\Vert Z),
\end{align*}
where $(a)$ stems from the well known result $D(X'\Vert X) = h(X) - h(X')$, see, e.g.,~\cite[p.~254]{covtho06}.
\end{proof}

We can now prove the achievability of $R^\perp(D)$.\\

\begin{thm}\label{thm:achievable}
\emph{For a source $X$ being an infinite length Gaussian random vector with zero mean,
$R^\perp(D)$ is achievable.}
\end{thm}

\begin{proof}
Let $X^{(N)}$ be the sub-vector containing the first $N$ elements of $X$.
For a fixed distortion $D=\tr(\boldsymbol{K}_{Z^{(N)}})/N$, the 
average mutual information per dimension $\frac{1}{N}I({X^{(N)}};{X^{(N)}}+{Z^{(N)}})$ is minimized when ${X^{(N)}}$ and ${Z^{(N)}}$ are jointly Gaussian and 
\begin{align}
\boldsymbol{K}_{Z^{(N)}} = \frac{1}{2}\sqrt{\boldsymbol{K}_{X^{(N)}}^2 + \alpha\boldsymbol{K}_{X^{(N)}}}-\frac{1}{2}\boldsymbol{K}_{X^{(N)}}, 
\end{align}
see~\cite{derost08}.
%
Let the $N$-dimensional shaped randomized lattice quantizer $\Q'_{N}$ be such that the dither is distributed as $-{{E'}^{(N)}}\sim\mathcal{U}(V'_0)$,
with $\bK_{E'^{(N)}}=\bK_{Z^{(N)}}$.
It follows that the coding rate of the quantizer is given by $R_{\Q_{N}} = \frac{1}{N}I({X^{(N)}};{X^{(N)}} + {{E'}^{(N)}})$.
The rate loss due to using $\Q_{N}$ to quantize ${X^{(N)}}$ is given by
\begin{align}
 R_{\Q_{N}}(D) - R^{\perp}(D) 
&=  \tfrac{1}{N}\Big[ I({X^{(N)}};{X^{(N)}}+{{E'}^{(N)}}) \nonumber\\
&\quad- I({X^{(N)}};{X^{(N)}}+{{E'}^{(N)}}_G)\Big] \nonumber\\
&\overset{(a)}{\leq} \tfrac{1}{N}D(f_{{{E'}^{(N)}}}(e)\|f_{{{E'_G}^{(N)}}}(e)),\label{eq:middle}
\end{align}
where $f_{{{E'_G}^{(N)}}}$ is the PDF of the Gaussian random vector ${{E'_G}^{(N)}}$, independent of ${{E'}^{(N)}}$ and ${X^{(N)}}$, and having the same first and second order statistics as ${{E'}^{(N)}}$.
In~\eqref{eq:middle}, inequality~$(a)$ follows directly from Lemma~\ref{lem:excessrate}, since the use of subtractive dither yields  the error ${{E'}^{(N)}}$ independent of ${X^{(N)}}$.

To complete the proof, we invoke Lemma~\ref{lem:shapedlattice}, which guarantees that the RHS of~\eqref{eq:middle} vanishes as $N\rightarrow \infty$.
\end{proof}

\begin{rem}
 \begin{enumerate}
  \item For zero mean stationary Gaussian random sources, $R^{\perp}(D)$ is achieved by taking $X$ in Theorem~\ref{thm:achievable} to be the complete input process. 
For this case, as shown in~\cite{derost08}, the Fourier transform of the autocorrelation function of $Z^{(N)}$ tends to the RHS of~\eqref{eq:Sz_def}.
\item For vector processes, the achievability of $R^{\perp}(D)$ follows by building $X$ in Theorem~\ref{thm:achievable} from the  concatenation of infinitely many consecutive vectors.
\item Note that if one has an infinite number of parallel scalar random processes, $R^{\perp}(D)$ can be achieved \emph{causally} by forming $X$ in Theorem~\ref{thm:achievable} from the $k$-th sample of each of the processes and using entropy coding after $\Q$.
 \end{enumerate}
\end{rem}

The fact that $R^{\perp}(D)$ can be realized causally is further illustrated in the following section.

\section{Realization of $R^{\perp}(D)$ by Causal Transform Coding }\label{sec:realiz_TC}
We will next show that for a Gaussian random vector $X\in\Rl^{N}$ with positive definite covariance matrix $\bK_{X}$, $R^{\perp}(D)$ can be realized by  \emph{causal} transform coding~\cite{habher74,pholin00}.
A typical transform coding architecture is shown in Fig.~\ref{fig:Causal_TCNF}.
In this figure, $\bT$ is an $N\times N$ matrix, and $W$ is a Gaussian vector, independent of $X$, with covariance matrix $\bK_{W}=\sigsq_{W}\bI$.
The system clearly satisfies the perfect reconstruction condition $Y=X + \bT^{-1}W$.
The reconstruction error is the Gaussian random vector $Z\eq Y-X$, and the MSE is 
$D=\frac{1}{N}\tr\set{\bK_{Z}}$, where 
$
\bK_{Z} = \sigsq_{W} \bT^{-1} \bT^{-T}
$.

\begin{figure}[htp]
 \centering
 \input{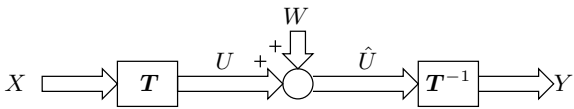}
 \caption{Transform coder.}
 \label{fig:Causal_TCNF}
\end{figure}

By restricting $\bT$ to be  lower triangular, the transform coder in Fig.~\ref{fig:Causal_TCNF} becomes causal, in the sense that $\forall k\in\set{1,..,N}$, the $k$-th elements of $U$ and $\hat{U}$ can be determined using just the first $k$ elements of $X$ and the $k$-th element of $W$.

To  have 
$
\frac{1}{N}I(X;Y)=R^{\perp}(D)
$,
 it is necessary and sufficient that 
%
\begin{align}
 \bT^{-1} \bT^{-T} = \bK_{Z^{\star}}/\sigsq_{W} \label{eq:bTbT_NF},
\end{align}
where the covariance matrix of the optimal distortion is~\cite{derost08}
\begin{align}
 \bK_{Z^{\star}}\eq 
\frac{1}{2}\hsqrt{\boldsymbol{K}_{X}^2 + \alpha\boldsymbol{K}_{X}}-\frac{1}{2}\boldsymbol{K}_{X}.\label{eq:Kzstardef2}
\end{align}

%
Since $\bT^{-1}$ is lower triangular,~\eqref{eq:bTbT_NF} is the Cholesky decomposition of $\bK_{Z^{\star}}/\sigsq_{W}$, which always exists.%
\footnote{Furthermore,  since $\bK_{Z^{\star}}>0$, there exists a unique $\bT$ having only positive elements on its main diagonal that satisfies~\eqref{eq:bTbT_NF}, see~\cite{horjoh85}.}
Thus, $R^{\perp}(D)$ can be realized by causal transform coding.

In practice, transform coders are implemented by replacing the (vector) AWGN channel $\hat{U}=V+W$ by a quantizer (or several quantizers) followed by entropy coding.
The latter process is simplified if the quantized outputs are independent.
When using quantizers with subtractive dither, this can be shown to be equivalent to having $\frac{1}{N}\sumfromto{k=1}{N}I(\hat{U}_{k}-W_{k};\hat{U}_{k})=\frac{1}{N} I(U;\hat{U})$
in the transform coder when using the AWGN channel.
Notice that, since $\bT$ in~\eqref{eq:bTbT_NF} is invertible, the  mutual information per dimension
$\frac{1}{N} I(U;\hat{U})$ is also equal to $R^{\perp}(D)$.
By the chain rule of mutual information we have
%
\begin{align}
\frac{1}{N}\sumfromto{k=1}{N}I(\hat{U}_{k}-W_{k};\hat{U}_{k}) 
 \geq \frac{1}{N} I(U;\hat{U}) = R^{\perp}(D),\label{eq:Ineq_key}
\end{align}
with equality iff the elements of $\hat{U}$ are mutually independent.
If $\hat{U}$ is Gaussian, this is equivalent to $\bK_{\hat{U}}$ being  diagonal.
Clearly, this cannot be obtained with the architecture shown in Fig.~\ref{fig:Causal_TCNF} using causal matrices (while at the same time satisfying~\eqref{eq:bTbT_NF}).
However, it can be achieved by using error feedback, as we show next.

Consider the scheme shown in Fig.~\ref{fig:Causal_TC}, where $\bA\in\Rl^{N\times N}$ is lower triangular and $\bF\in\Rl^{N\times N}$ is strictly lower triangular.
\begin{figure}[htp]
 \centering
 \input{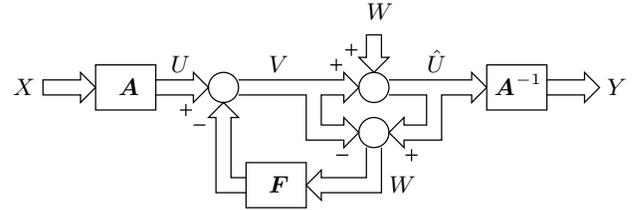}
 \caption{A causal transform coding scheme with error feedback.}
 \label{fig:Causal_TC}
\end{figure}
Again, a sufficient and necessary condition to have $\frac{1}{N}I(X;Y)=R^{\perp}(D)$ is that $\bK_{Z}=\bK_{Z^{\star}}$, see~\eqref{eq:Kzstardef2}, i.e.,
%
\begin{align}
\sigsq_{W} \bA^{-1}(\bI-\bF)\left[\bA^{-1}(\bI-\bF)\right]^{T} = \bK_{Z^{\star}}\nonumber\\
\iff
(\bI-\bF)(\bI-\bF)^{T} = \bA \bK_{Z^{\star}} \bA^{T}/\sigsq_{W}. \label{eq:I_F_I_F}
\end{align}
%
%
On the other hand, equality in~\eqref{eq:Ineq_key} is achieved only if
\begin{align}
 \bK_{\hat{U}} = \bA\bK_{X}  \bA^{T}+  \sigsq_{W}(\bI - \bF)(\bI - \bF)^{T}= \bD,\label{eq:bI}
\end{align}
for some diagonal matrix $\bD$ with positive elements.
If we substitute the Cholesky factorization $\bK_{Z^{\star}}=\bL \bL^{T}$ into~\eqref{eq:I_F_I_F}, we obtain
$
 (\bI-\bF)(\bI-\bF)^{T} = \bA \bL\bL^{T}\bA^{T}/\sigsq_{W}
$,
and thus
%
\begin{align}
 \bA = \sigma_{W}(\bI-\bF)\bL^{-1}.\label{eq:bA_short}
\end{align}
Substituting the above into~\eqref{eq:bI} we obtain
%
\begin{align}
 \bD = \sigsq_{W}(\bI-\bF)\left[ \bL^{-1}\bK_{X}\bL^{-T} + \bI\right](\bI-\bF)^{T}\label{eq:the_one}
\end{align}
%
Thus, there exist%
\footnote{For any positive definite matrices $\bK_{X}$ and $\bK_{Z^{\star}}=\bL\bL^{T}$, there exists a \emph{unique} matrix $\bF$ having zeros on its main diagonal that satisfies~\eqref{eq:the_one}, see~\cite{derque08}.
} 
$\bA$ and $\bF$ satisfying~\eqref{eq:I_F_I_F} and~\eqref{eq:bI}.
Substitution of~\eqref{eq:bA_short} into~\eqref{eq:the_one} yields
$
 \bD = \bA\left(\bK_{X} + \bK_{Z^{\star}}\right) \bA^{T}
$,
and 
$ 
\log \abs{\bD} = 2\log\abs{\bA} + \log\abs{\bK_{x}  +  \bK_{Z^{\star}}}
$.
From~\eqref{eq:I_F_I_F} and the fact that $\abs{\bI-\bF}=1$ it follows that
$
 \abs{\bA}^{2} = \sigsq_{W} / \abs{\bK_{Z^{\star}}}
$,
and therefore%
\footnote{The last equality in~\eqref{eq:esta} follows from the expression for $R^{\perp}(D)$ for Gaussian vector sources derived in~\cite{derost08}.}
%
\begin{align}
\tfrac{1}{N}&\sumfromto{k=1}{N}I(V_{k};\hat{U}_{k}) 
=
\tfrac{1}{N}\sumfromto{k=1}{N}\log\Big(\tfrac{\sigsq_{\hat{U}_{k}}}{\sigsq_{W}}\Big)
=\tfrac{1}{2N} \log \tfrac{\abs{\bD}}{\sigsq_{W}} \nonumber\\
&=  \tfrac{1}{2N}\log\abs{\bK_{x}  +  \bK_{Z^{\star}}}  - \tfrac{1}{2N}\log\abs{ \bK_{Z^{\star}}}\nonumber\\
&= \tfrac{1}{2N}\!\sumfromto{k=1}{N}\!\log 
\left( 
\tfrac{\hsqrt{\lambda_{k}^{2} +\lambda_{k}\alpha } + \lambda_{k}}{ \hsqrt{\lambda_{k}^{2} +\lambda_{k}\alpha } - \lambda_{k}}
\right)
= R^{\perp}(D),\label{eq:esta}
\end{align}
thus achieving equality in~\eqref{eq:Ineq_key}.

We have seen that the use of error feedback allows one to make the average scalar mutual information between the input and output of each AWGN channel in the transform domain equal to $R^{\perp}(D)$.
In the following section we show how this result can be extended to stationary Gaussian processes.

\section{Achieving $R^{\perp}(D)$ by Noise Shaping}\label{sec:noise_shap}
In this section we show that, for any colored stationary Gaussian stationary source and for any positive distortion, $R^{\perp}(D)$ can be realized by noise shaping, and that $R^{\perp}(D)$ is achievable using \emph{memory-less} entropy coding.

\subsection{Realization of $R^{\perp}(D)$ by Noise-Shaping}
The fact that $R^{\perp}(D)$ can be realized by the additive colored  Gaussian noise test channel of Fig.~\ref{fig:forwartdtc} suggests that $R^{\perp}(D)$ could also be achieved by an 
\emph{additive white Gaussian noise} (AWGN) channel embedded in a noise-shaping feedback loop, see Fig.~\ref{fig:Block_diag_NSDPCM}.
In this figure,
$\set{X_{k}}$ is a Gaussian stationary process with PSD $S_{x}\ejw$.
The filters $A(z)$ and $F(z)$ are LTI.
The AWGN channel is situated between $V$ and $\hat{U}$, where white Gaussian noise $\set{W_{k}}$, independent of $\set{X_{k}}$,  is added.
The reconstructed signal $Y$ is obtained by passing $\hat{U}$ through the filter $A(z)^{-1}$, yielding the reconstruction error 
$Z_{k}= Y_{k}-X_{k}$.

\begin{figure}[htp]
 \centering
 \input{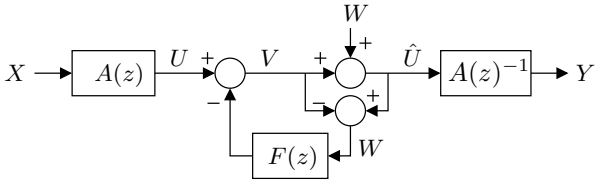}
 \caption{Test channel built by embedding the AWGN channel $\hat{U}_{k}= V_{k}+W_{k}$ in a noise feedback loop.}
 \label{fig:Block_diag_NSDPCM}
\end{figure}

The following  theorem states that, for this scheme,  the \emph{scalar} mutual information across the AWGN channel can actually equal 
$R^{\perp}(D=\sigsq_{Z})$.
%
%
%
\begin{thm}\label{thm:Realizable_FQ}
\emph{
Consider the scheme in Fig.~\ref{fig:Block_diag_NSDPCM}.
Let $\set{X_{k}}$,  $\set{W_{k}}$ be independent stationary Gaussian random processes.
Suppose that
 the differential entropy rate of $\set{X_{k}}$ is bounded,
and that
$\set{W_{k}}$ is white.
Then, for every $D>0$, there exist causal and stable filters $A(z)$, $A(z)^{-1}$ and $F(z)$ such that
%
\begin{align}
 I(V_{k};\hat{U}_{k}) = R^{\perp}(D), \textrm{ where $D\eq \sigsq_{Z}$.}\label{eq:scalar_I_equals_Rperp}
\end{align}
}
\end{thm}

%
%
\begin{proof}
Consider all possible choices of the filters $A(z)$ and $F(z)$ such that
the obtained sequence $\set{\hat{U}_{k}}$ is white, i.e., such that $S_{\hat{U}}\ejw=\sigsq_{\hat{U}},\,\forallwinpipi$.
From Fig.~\ref{fig:Block_diag_NSDPCM}, this is achieved iff the filters $A(z)$ and $F(z)$ satisfy
%
\begin{align}
 \sigsq_{\hat{U}} = \abs{A\ejw}^{2}S_{X}\ejw + \abs{1-F\ejw}^{2}\sigsq_{W}.\label{eq:uhat_whitwe}
\end{align}
On the other hand, since $\set{W_{k}}$ is Gaussian, a necessary and sufficient condition in order to achieve
$R^{\perp}(D)$ is that 
%
\begin{align}
 S_{Z}\ejw 
&= \abs{1-F\ejw}^{2} \abs{A\ejw}^{-2} \sigsq_{W}\label{eq:Sz}\\
&=\frac{1}{2}\!\left(\! \hsqrt{S_X(\omega)+\alpha} - \hsqrt{S_X(\omega)}\right)\!\hsqrt{S_X(\omega)} \\
& \eq S_{Z^{\star}}\ejw,\quad \forallwinpipi.\label{eq:Szstar2}
\end{align}
This holds iff
$
\abs{A\ejw}^{2} = \sigsq_{W}\abs{1-F\ejw}^{2}/ S_{Z^{\star}}\ejw
$.
Substituting the latter and~\eqref{eq:Szstar2} into~\eqref{eq:uhat_whitwe}, and after some algebra, we obtain
%
\begin{subequations}\label{eq:opt_filters}
\begin{align}
\!\!\!\abs{1\!-\!F\ejw}^{2} \!\!
&= 
\frac{\sigsq_{\hat{U}}}{\sigsq_{W}}\!\!
\left[ 
\frac
{
\!\hsqrt{S_{X}\ejw \!+\! \alpha} -\! \hsqrt{S_{X}\ejw}
}
{\hsqrt{\alpha}}
\right]^{2},
\label{eq:uno_min_F_opt} 
\\
\abs{A\ejw}^{2} 
&=
2\sigsq_{\hat{U}}
\frac{\hsqrt{S_{X}\ejw \!+\! \alpha} -\! \hsqrt{S_{X}\ejw}}{\alpha \hsqrt{S_{X}\ejw}}\label{eq:Aopt}.
\end{align}
\end{subequations}
%
Notice that the functions on the right hand sides of~\eqref{eq:opt_filters} are bounded and positive for all $\w\in\pipi$, and that
a bounded differential entropy rate of $\set{X_{k}}$ implies that $|\int_{-\pi}^{\pi}S_{X}\ejw d\w|<\infty$.
From the Paley-Wiener criterion~\cite{wiepal34} (see also, e.g.,~\cite{proman96}), this implies that $(1-F(z))$, $A(z)$ and $A(z)^{-1}$ can be chosen to be stable and causal.
Furthermore, recall that for any fixed $D>0$, the corresponding value of $\alpha$ is unique (see~\cite{derost08}), and thus fixed.
Since the variance $\sigsq_{W}$ is also fixed, it follows that each frequency response magnitude $\abs{1-F\ejw}$ that satisfies~\eqref{eq:uno_min_F_opt} can be associated to a unique value of
$\sigsq_{\hat{U}}$.
Since $F(z)$ is strictly causal and stable, the minimum value of the variance $\sigsq_{\hat{U}}$ is achieved when
%
\begin{align}
 \intpipi{\log\abs{1-F\ejw}} =0,\label{eq:Bode}
\end{align}
i.e., if $1-F(z)$ has no zeros outside the unit circle (equivalently, if $1-F(z)$ is minimum phase), see, e.g.,~\cite{serbra97}.
If we choose in~\eqref{eq:uno_min_F_opt} a filter $F(z)$ that satisfies~\eqref{eq:Bode}, and then we take the logarithm and integrate both sides of~\eqref{eq:uno_min_F_opt}, we obtain 
%
\begin{align*}
\frac{1}{2}&
\!\log\!\left(\frac{\sigsq_{\hat{U}}}{\sigsq_{W}}\!\right)
 =
\frac{1}{2\pi}\!\Intfromto{-\pi}{\pi}{\!\log
\left[\frac{\hsqrt{\alpha}}{\hsqrt{S_{X}\ejw \!+\! \alpha} - \hsqrt{S_{X}\ejw}}
\right] }d\w\\
&=
\frac{1} {2\pi}\!
\Intfromto{-\pi}{\pi} {\!\log\!
\left[
\frac
{\hsqrt{S_{X}\ejw \!+\! \alpha} + \hsqrt{S_{X}\ejw}}
{\hsqrt{\alpha}}
\right] }d\w
=
R^{\perp}(D).
\end{align*}
where~\eqref{eq:RperpProcDef} has been used. 
We then have that 
\begin{align*}
R^{\perp}(D) 
&\overset{\hphantom{(a)}}{=} \frac{1}{2}\log\Big(\frac{\sigsq_{\hat{U}}}{\sigsq_{W}}\Big) 
= \frac{1}{2}\log(2\pi\expo{}\sigsq_{\hat{U}}) -\frac{1}{2}\log(2\pi\expo{}\sigsq_{W}) \\
&\overset{(a)}{=} h(\hat{U}_{k}) - h(W_{k})\\
&\overset{(b)}{=} h(\hat{U}_{k}) - h(V_{K}+ W_{k}| V_{k}  )
= I(V_{k};\hat{U}_{k}),
\end{align*}
where $(a)$ follows from the Gaussianity of $W_{k}$ and $\hat{U}_{k}$, 
and 
$(b)$ from the fact that $W_{k}$ is independent of $V_{k}$ (since $F$ is strictly causal).
This completes the proof.
Alternatively,
%
%
%
\begin{align*}
 R^{\perp}(&D)
\overset{(a)}{\leq} 
\Irate{X}{Y}\\
&\overset{\hphantom{(a)}}{=}
\bar{h}(A^{-1}\set{\hat{U}_{k}}) 
- 
\bar{h}(\set{X_{k}} + A^{-1}(1-F)\set{W_{k}}|\set{X_{k}})\\
&\overset{\hphantom{(b)}}{=}
\bar{h}(A^{-1}\set{\hat{U}_{k}})  
- 
\bar{h}( A^{-1}(1-F)\set{W_{k}}  )\\
&\overset{(b)}{=}
\bar{h}(\set{\hat{U}_{k}})  
- 
\bar{h}((1-F)\set{W_{k}}  )\\
&\overset{(c)}{\leq}
h(\hat{U}_{k}|U_{k}^{-})  
- 
h(W_{k}) \overset{(d)}{\leq}
h(\hat{U}_{k})  
- 
h(W_{k})\\
&\overset{(e)}{=} h(\hat{U}_{k}) - h(V_{K}+ W_{k}| V_{k}  )
= I(V_{k};\hat{U}_{k}),
\end{align*}
In~$(a)$, equality is achieved iff the right hand side of~\eqref{eq:Sz} equals~\eqref{eq:uno_min_F_opt}, i.e., if $Z$ has the optimal PSD.
Equality~$(b)$ holds because $|{\int_{-\pi}^{\pi}\log\abs{A\ejw}}|d\w<\infty$, which follows from~\eqref{eq:Aopt}.
The fact that $\set{\hat{U}_{k}}$ is stationary has been used in~$(c)$, wherein equality is achieved iff $\abs{1-F}$ is minimum phase, i.e., if~\eqref{eq:Bode} holds.
Equality in~$(d)$ holds if an only if the elements of $\set{\hat{U}_{k}}$ are independent, which, from the Gaussianity of $\set{\hat{U}_{k}}$, is equivalent to~\eqref{eq:uhat_whitwe}.
Finally, $(e)$ stems from the fact that
$W_{k}$ is independent of $V_{k}$.
\end{proof}
Notice that the key to the proof of Theorem~\ref{thm:Realizable_FQ} relies on knowing a priori the PSD of the end to end distortion required to realize $R^{\perp}(D)$. 
Indeed, one could also use this fact to realize $R^{\perp}(D)$ by embedding the AWGN in a DPCM feedback loop, and then following a reasoning similar to that in~\cite{zamkoc08}.

\subsection{Achieving $R^{\perp}(D)$ Through Feedback Quantization}
In order to achieve $R^{\perp}(D)$ by using a quantizer instead of an AWGN channel, one would require the quantization errors to be Gaussian. 
This cannot be achieved with scalar quantizers.
However, as we have seen in~\ref{sec:background}, dithered lattice quantizers are able to yield quantization errors approximately Gaussian as the lattice dimension tends to infinity.
The sequential (causal) nature of the feedback architecture does not immediately allow for the possibility of using vector quantizers. However, if several sources are to be processed simultaneously, we can overcome this difficulty by using an idea suggested in~\cite{zamkoc08} where the sources are processed in parallel by separate feedback quantizers. 
The feedback quantizers are operating independently of each other except that their scalar quantizers are replaced by a single vector quantizer. 
If the number of parallel sources is large, then the vector quantizer guarantees that the marginal distributions of the individual components of the quantized vectors becomes approximately Gaussian distributed. Thus, due to the dithering within the vector quantizer, each feedback quantizer observes a sequence of i.i.d.\ Gaussian quantization noises. 
Furthermore, the effective coding rate (per source) is that of a high dimensional entropy constrained dithered quantizer (per dimension).

The fact that the scalar mutual information between $V_{k}$ and $\hat{U}_{k}$ 
equals the mutual information rate between $\set{V_{K}}$ and $\set{\hat{U}_{k}}$
in each of the parallel coders  implies that $R^{\perp}(D)$ can be achieved by using a memoryless entropy coder.

\section{Rate Loss with Dithered Feedback Quantization}
The results presented in  sections~\ref{sec:realiz_TC} and~\ref{sec:noise_shap} suggest that 
if a test channel embedding an AWGN channel realizes $R^{\perp}(D)$, then a source coder obtained by replacing the AWGN channel by a dithered, finite dimensional lattice quantizer, would exhibit a rate close to $R^{\perp}(D)$.

The next theorem, whose proof follows the line of the results given in~\cite[sec.~VII]{zamkoc08}, provides an upper bound on the rate-loss incurred in this case.
\begin{thm}
\emph{
Consider a source coder with a finite dimensional subtractively dithered lattice quantizer $\Q$.
If when replacing the quantizer by an AWGN channel  the scalar mutual information across the channel equals $R^{\perp}(D)$, then
the scalar entropy of the quantized output exceeds $R^{\perp}(D)$ by at most  $0.254$ bit/dimension.
}
\end{thm}

\begin{proof}
Let $W$ be the noise of the AWGN channel, and $V$ and $\hat{U}$ denote the channel input and output signals.
From the conditions of the theorem, we have that
%
\begin{align}
 I(V_{k};\hat{U}_{k}) = R^{\perp}(D).\label{eq:Iscalar_eq_Rperp}
\end{align}
If we now replace the AWGN by a dithered quantizer with subtractive dither $\nu$, such that the quantization noise $W'$ is obtained with the same first and second order statistics as $W$, then the end to end MSE remains the same.
The corresponding signals in the quantized case, namely $V'$ and $\hat{U}'$, will also have the same second order statistics as their Gaussian counterparts $V$ and $\hat{U}$.
Thus, by using Lemma~\ref{lem:excessrate} we obtain
%
\begin{align}
 I(V_{k}';\hat{U}_{k}') 
&\overset{ }{\leq} R^{\perp}(D) + D(\hat{U}_{k}'\Vert \hat{U}_{k})\label{eq:I_rate_loss}.
\end{align}
Finally, from~\cite[Theorem~1]{zamfed92}, we have that
%
$
 H(\Q(V_{k}+\nu_{k})|\nu_{k} )= I(V_{k}';\hat{U}_{k}')
$.
Substitution of~\eqref{eq:I_rate_loss} into this last equation yields the result. 
\end{proof}

\section{Conclusions}
We have proved the achievability of $R^{\perp}(D)$ by using lattice quantization with subtractive dither.
We have shown that $R^{\perp}(D)$ can be realized causally, and that the use of feedback allows one to achieve $R^{\perp}(D)$ by using memoryless entropy coding. 
We also showed that the scalar entropy of the quantized output when using optimal finite-dimensional dithered lattice quantization exceeds $R^{\perp}(D)$ by at most $0.254$ bits/dimension.
\balance
%
%
\bibliographystyle{\BibPath/IEEEtran}

\end{document}